\documentclass[aps,pra,twocolumn,superscriptaddress,floatfix,nofootinbib,showpacs,longbibliography]{revtex4-1}
\usepackage{lipsum, babel}
\usepackage{amsmath}
\usepackage{float}
\usepackage{accents}
\usepackage{amsmath}
\usepackage[utf8]{inputenc}  
\usepackage[T1]{fontenc}     
\usepackage[sc,osf]{mathpazo}
\usepackage[scaled=0.86]{berasans}  
\usepackage[colorlinks=true, citecolor=blue, urlcolor=blue]{hyperref}
\makeatletter
\newcommand{\setword}[2]{%
  \phantomsection
  #1\def\@currentlabel{\unexpanded{#1}}\label{#2}%
}
\makeatother
\usepackage{graphicx} 
\usepackage{amsmath}
\usepackage{comment}
\usepackage{nicematrix,tikz}
\newcommand{\tens}[1]{%
  \mathbin{\mathop{\otimes}\limits_{#1}}%
}
\usepackage{nicematrix,tikz}
\usepackage[babel]{microtype}  
\usepackage{amsmath,amssymb,amsthm,bm,amsfonts,mathrsfs,bbm} 

\usepackage{xspace}  
\usepackage{pgf,tikz}
\usepackage{xcolor}
\usepackage{multirow}
\usepackage{array}
\usepackage{bigstrut}
\usepackage{braket}
\usepackage{color}
\usepackage{natbib}
\usepackage{multirow}
\usepackage{mathtools}
\usepackage{float}
\usepackage[caption = false]{subfig}
\usepackage{xcolor,colortbl}
\usepackage{color}
\usepackage{diagbox}
\usepackage{enumitem}
\setitemize{wide}

\newcommand{\be}{\begin{equation}}
\newcommand{\ee}{\end{equation}}
\newcommand{\ba}{\begin{eqnarray}}
\newcommand{\ea}{\end{eqnarray}}
\newcommand{\ketbra}[2]{|#1\rangle \langle #2|}
\newcommand{\tr}{\operatorname{Tr}}

\newtheorem{theorem}{Theorem}
\newtheorem{corollary}{Corollary}
\newtheorem{definition}{Definition}
\newtheorem{proposition}{Proposition}

\def\>{\rangle}
\def\<{\langle}

\usepackage{centernot}
\usepackage{subfig}

\begin{document}

\title{Timelike correlations and quantum tensor product structure}

\author{Samrat Sen}
\affiliation{Department of Physics of Complex Systems, S.N. Bose National Center for Basic Sciences, Block JD, Sector III, Salt Lake, Kolkata 700106, India.}

\author{Edwin Peter Lobo}
\affiliation{School of Physics, IISER Thiruvananthapuram, Vithura, Kerala 695551, India.}

\author{Ram Krishna Patra}
\affiliation{Department of Physics of Complex Systems, S.N. Bose National Center for Basic Sciences, Block JD, Sector III, Salt Lake, Kolkata 700106, India.}

\author{Sahil Gopalkrishna Naik}
\affiliation{Department of Physics of Complex Systems, S.N. Bose National Center for Basic Sciences, Block JD, Sector III, Salt Lake, Kolkata 700106, India.}

\author{Anandamay Das Bhowmik}
\affiliation{Physics and Applied Mathematics Unit, Indian Statistical Institute, 203 BT Road, Kolkata, India.}

\author{Mir Alimuddin}
\affiliation{Department of Physics of Complex Systems, S.N. Bose National Center for Basic Sciences, Block JD, Sector III, Salt Lake, Kolkata 700106, India.}

\author{Manik Banik}
\affiliation{Department of Physics of Complex Systems, S.N. Bose National Center for Basic Sciences, Block JD, Sector III, Salt Lake, Kolkata 700106, India.}

\begin{abstract}
The state space structure for a composite quantum system is postulated among several mathematically consistent possibilities that are compatible with local quantum description. For instance, unentangled Gleason's theorem allows a state space that includes density operators as a proper subset among all possible composite states. However, bipartite correlations obtained in Bell type experiments from this broader state space are in-fact quantum simulable, 
and hence such spacelike correlations are no good to make distinction among different compositions. In this work we analyze communication utilities of these different composite models and show that they can lead to distinct utilities in a simple communication game involving two players. Our analysis, thus, establishes that beyond quantum composite structure can lead to beyond quantum correlations in timelike scenario and hence welcomes new principles to isolate the quantum correlations from the beyond quantum ones. We also prove a no-go that the classical information carrying capacity of different such compositions cannot be more than the corresponding quantum composite systems.   
\end{abstract}

\maketitle	

\section{Introduction}
The tensor product postulate of quantum mechanics, also cited as the `zeroth’ axiom in literature \cite{Zurek09}, describes the Hilbert space of a composite system to be the tensor product of the components’ Hilbert spaces \cite{Neumann55,Dirac66,Jauch68}. A recent study, however, logically derives this postulate from the state postulate and the measurement postulate rather than taking it as an independent one \cite{Carcassi21}. Nevertheless, within this tensor product structure, the unentangled Gleason's theorem assigns state spaces for the composite systems that include density operators (the quantum states) as a proper subset \cite{Klay87,Wallach00,Barnum05}. In fact, assuming individual systems' description to be quantum, several mathematical models are possible for the composite state and effect spaces that yield consistent outcome probability. The framework of generalized probability theory (GPT) \cite{Hardy01,Barrett07,Chiribella11,Barnum11,Masanes11,Plavala21} is well-suited to study these different composite models. Physical constraints, such as no-signaling and local tomography, limit the composite state spaces to be constrained within two extremes \cite{Namioka69,Barker76,Barker81,Auburn20}-- the minimal tensor product composition containing only separable states and the maximal tensor product composition containing beyond quantum states that are positive on product tests (POPT) and compatible with unentangled Gleason's theorem. The corresponding effect spaces are specified in accordance with the `no-restriction’ hypothesis \cite{Chiribella10} that includes all the mathematically consistent effects in the theory. 

A natural question is whether these different composite models can lead to stronger than quantum correlations that in turn will make them distinct from quantum state space and will put an embargo on their physical existence. For the bipartite case a negative response comes through the work of Barnum {\it et al.} \cite{Barnum10} which states that no such composition can produce any beyond quantum spacelike correlations in Bell type experimental scenario. While maximal composition involving more than two subsystems can yield stronger than quantum correlation in typical Bell scenario \cite{Acin10}, recently in Ref.\cite{Lobo21} it has been shown that even in bipartite case stronger than quantum correlations are possible if the typical classical input-classical output Bell scenario is generalized to quantum input-classical output  semi-quantum scenario \cite{Buscemi12}. Although this quantum input scenario disallows all the compositions having beyond entangled states it requires trustworthy verifiers in producing some predetermined unentangled quantum inputs \cite{Lobo21}. In a completely different approach, recently the authors in \cite{Naik22} have shown that the bipartite minimal composition can yield stronger than quantum correlation if timelike scenario is considered. More particularly, it has been shown that a communication game played between two timelike separated players - a sender, and a receiver at the sender's causal future -- cannot be won perfectly by communicating two elementary quantum systems (qubits) if the composite state space is considered to be the standard quantum one, whereas the game becomes perfectly winnable if the composition is assumed to be the minimal one. Although the minimal composition consisting of separable states only cannot produce any nonlocal correlation in Bell like scenario, existence of beyond quantum effects in this theory results in beyond quantum correlations in timelike scenario. This result may lead to an impression that beyond quantum effects are necessary to obtain beyond quantum timelike correlations. In this work we, however, show that such an intuition is, in-fact, not true. More particularly, we show that maximal composition that allows only product effects but permits beyond quantum states can also yield beyond quantum correlations in timelike scenario. Thus, while Barnum {\it et al.} result \cite{Barnum10} shows that spacelike correlations are no good to establish the beyond quantum nature of the bipartite maximal composition, our result establishes that timelike correlations do serve the purpose here. We then proceed to prove a no-go that although the maximal composition allows beyond quantum timelike correlations, the classical information carrying capacity of such models cannot be more than the corresponding quantum composite systems. In fact, we prove a generic result regarding the information capacity of composite systems in the GPT framework.  

\section{Preliminaries}
\subsection{Framework of GPT}
We start by briefly recalling the framework of GPT. For a detailed overview of this framework we refer to the works \cite{Hardy01,Barrett07,Chiribella11,Barnum11,Masanes11,Plavala21}. In the recent past several interesting results have been reported within this framework \cite{Muller2012,Banik2015,Banik2019,Bhattacharya2020,Saha2020,Saha2021,Mayalakshmi2022}. A GPT is specified by a list of system types and the composition rules specifying combination of several systems, where a system $S$ is specified by identifying the three-tuple $\left(\Omega_S,\mathcal{E}_S,\mathcal{T}_S\right)$ of the state space, effect space, and the set of transformations. In a prepare and measure scenario, which will be considered in this work, it is sufficient to describe $\Omega_S$ and $\mathcal{E}_S$ only.  

{\it State space [$\Omega_S$]:} A state $\omega_S$ for a system $S$ is a mathematical object that yields outcome probabilities for all the measurements that can possibly be carried out on the system. Collection of all allowed states form the state space $\Omega_S$, and generally it is considered to be a compact-convex set embedded in some real vector space $V$. Convexity assures the fact that if $\omega_1$ and $\omega_2$ are allowed states then their classical mixture $p\omega_1+(1-p)\omega_2$ is also a valid state. On the other hand, compactness assures that there is no physical distinction between states that can be prepared exactly, and states that can be prepared to arbitrary accuracy \cite{Krumm17}. The extreme points of the set $\Omega_S$ are called pure states or states of maximal knowledge. 

{\it Effect space [$\mathcal{E}_S$]:} An effect $e$ is a linear functional acting on $V$ such that $e:\Omega_S\to[0,1]$. The unit effect is defined by $u(\omega)=1,~\forall~\omega\in\Omega_S$. The set of all proper effects $\mathcal{E}_S\equiv\{e~|~0\le e(\omega)\le1,~\forall~\omega\in\Omega_S\}$ is the convex hull of zero effect, unit effect and the extremal effects and embedded in the vector space $V^\star$ dual to $V$. A measurement $\mathcal{M}$ is a collection of effects that sum to the unit effect, {\it i.e.} $\mathcal{M}\equiv\{e_i\in\mathcal{E}_S~|~\sum_ie_i=u\}$. 

{\it State and effect cones:} Sometime it is mathematically convenient to work with the notion of unnormalized states and effects. The set of unnormalized states $V_+\subset V$ is the conical hull of $\Omega_s$, {\it i.e.}, $r\omega\in V_+$ for $r\ge0$ and $\omega\in\Omega_S$. The set of unnormalized effects is its dual cone $V_+^\star\subset V^\star$, {\it i.e.},  $V_+^\star\equiv\left\{e~|~e(\omega)\ge0,~\forall~\omega\in V_+\right\}$. The formulation generally assumes the `no-restriction hypothesis' which demands that the state and effect cones are dual to each other \cite{Chiribella11}.

{\it Composite system:} Given two systems with state spaces $\Omega_A\subset V_A$ and $\Omega_B\subset V_B$, the state space $\Omega_{AB}$ for the composite systems is embedded in the vector space $V_{AB}$ which is the tensor product of the component vector spaces, {\it i.e.} $V_{AB}=V_A\otimes V_B$ \cite{Carcassi21}. Although the choice of $\Omega_{AB}$ is not unique, the no signaling principle and tomographic locality postulate \cite{Hardy01} bound the choices within two extremes -- the minimal tensor product space and maximal tensor product space \cite{Namioka69}. More formally,
\begin{align*}
\Omega_{AB}^{\min}&\equiv\{\omega_{AB}=\sum_ip_i\omega_A^i\otimes\omega_B^i~|~\omega_A^i\in\Omega_A,\\
&~~~~~~~\omega_B^i\in\Omega_B;~p_i\ge0~\&~\sum_i p_i=1\},\\
\Omega_{AB}^{\max}&\equiv\left\{\omega_{AB}~\in V_{AB}~|~1\ge~e_A\otimes e_B(\omega_{AB})\ge0,~~~~~~~~~~~~~~~~\right.\\
&\left.~~~~~~~~~~~~~~~~~~~~~~\forall~e_A\in\mathcal{E}_A~\&~e_B\in\mathcal{E}_B\right\}.
\end{align*}
It is not hard to see that the cone $(V^{\min}_{AB})_+$ is isomorphic to the dual cone $(V^{\max}_{AB})_+$. Therefore, in accordance with the no-restriction hypothesis for the case of minimal composition, the effect cone $(V^{\min}_{AB})^\star_+ \cong (V^{\max}_{AB})_+$, and for the case of maximal composition, the effect cone $(V^{\max}_{AB})^\star_+ \cong (V^{\min}_{AB})_+$. The symbol $\cong$ denotes isomorphism.

\subsection{Quantum theory: a GPT} 
Quantum theory can be seen as a special instance of a GPT. State space of a $d$-level quantum system associated with complex Euclidean space $\mathbb{C}^d$ is the set of density operators acting on $\mathbb{C}^d$, {\it i.e.}, $\Omega(\mathbb{C}^d)\equiv\mathcal{D}(\mathbb{C}^d)$. The set $\mathcal{D}(\mathbb{C}^d)$ is a convex compact set embedded in $\mathbb{R}^{d^2-1}$. The unnormalized state cone is the set of all non-negative operators $\mathcal{P}(\mathbb{C}^d):=\{\lambda\rho~|~\lambda\ge0~\&~\rho\in\mathcal{D}(\mathbb{C}^d)\}$, which is also the unnormalized effect cone. In other words, quantum theory is self dual. The minimal composition of two quantum systems associated with Hilbert spaces $\mathbb{C}^{d_A}$ and $\mathbb{C}^{d_B}$ allows only separable state we call it as SEP composition and denote the resulting system as the triplet $S^{AB}_{\text{SEP}}\equiv\left[\mathbb{C}^{d_A},\mathbb{C}^{d_B},\bigotimes_{\text{SEP}}\right]$. Formally the state space for the SEP composition is given by
\begin{align*}
\Omega_{\text{SEP}}(\mathbb{C}^{d_A},\mathbb{C}^{d_B}):=\left\{\rho_{AB}=\sum_ip_i\rho_A^i\otimes\rho^i_B~|~p_i\ge0\right.\\
\left.\&~\sum_ip_i=1;~\rho^i_A\in\mathcal{D}(\mathbb{C}^{d_A}),~\rho^i_B\in\mathcal{D}(\mathbb{C}^{d_B})\right\}.
\end{align*}
Since $\Omega_{\text{SEP}}(\mathbb{C}^{d_A},\mathbb{C}^{d_B})$ contains only separable states, the corresponding effect space $\mathcal{E}_{\text{SEP}}(\mathbb{C}^{d_A},\mathbb{C}^{d_B})$ contains effects that are not allowed in quantum theory. Entanglement witness operators yielding positive probability on separable states are valid effects in this composition although they are not allowed in quantum theory. On the other extreme, the maximal composition, which we will call $\overline{\text{SEP}}$, and the resulting system denote as $S^{AB}_{\overline{\text{SEP}}}\equiv\left[\mathbb{C}^{d_A},\mathbb{C}^{d_B},\bigotimes_{\overline{\text{SEP}}}\right]$, has the state space
\begin{align*}
\Omega_{\overline{\text{SEP}}}(\mathbb{C}^{d_A},\mathbb{C}^{d_B}):=\left\{W_{AB}\in\text{Herm}(\mathbb{C}^{d_A}\otimes\mathbb{C}^{d_B})~|~\right.\\
\left.\tr(W_{AB})=1, \tr[W_{AB}(\pi_A\otimes\pi_B)]\ge0~ \right.\\
\left.\forall~\pi_A\in\mathcal{P}(\mathbb{C}^{d_A}),\pi_B\in\mathcal{P}(\mathbb{C}^{d_B})\right\}.
\end{align*}
Here $\text{Herm}(\mathcal{X})$ denotes the set of Hermitian operators acting on the space $\mathcal{X}$, and normalization demands $\tr(W_{AB})=1~\forall~W_{AB}\in\Omega_{\overline{\text{SEP}}}(\mathbb{C}^{d_A},\mathbb{C}^{d_B})$. The unnormalized effect cone corresponding to $\mathcal{E}_{\overline{\text{SEP}}}(\mathbb{C}^{d_A},\mathbb{C}^{d_B})$ is identical to the unnormalized state cone corresponding to the set $\Omega_{\text{SEP}}(\mathbb{C}^{d_A},\mathbb{C}^{d_B})$. For the quantum case $S^{AB}_{Q}\equiv\left[\mathbb{C}^{d_A},\mathbb{C}^{d_B},\bigotimes_{Q}\right]$ we have $\Omega_{Q}(\mathbb{C}^{d_A},\mathbb{C}^{d_B})=\mathcal{D}(\mathbb{C}^{d_A}\otimes\mathbb{C}^{d_B})$, and the effect cone is identical to the state cone which represents the self duality of quantum theory. The following set inclusion relations are immediate
\begin{align*}
\Omega_{\text{SEP}}(\mathbb{C}^{d_A},\mathbb{C}^{d_B})\subset\Omega_{Q}(\mathbb{C}^{d_A},\mathbb{C}^{d_B})\subset\Omega_{\overline{\text{SEP}}}(\mathbb{C}^{d_A},\mathbb{C}^{d_B}),\\ 
\mathcal{E}_{\overline{\text{SEP}}}(\mathbb{C}^{d_A},\mathbb{C}^{d_B})\subset\mathcal{E}_{Q}(\mathbb{C}^{d_A},\mathbb{C}^{d_B})\subset\mathcal{E}_{\text{SEP}}(\mathbb{C}^{d_A},\mathbb{C}^{d_B}).
\end{align*}
In between SEP and $\overline{\text{SEP}}$, many other compositions can be defined by appending/deducting suitable states/effects. Among these, quantum composition is the only one that is self dual.    

\subsection{Operational notions of dimension}
The dimension of the vector space $V$ in which the set $\Omega_S$ is embedded is a well defined concept, but it does not carry any operational signature. However, operationally motivated notion of dimension can be defined through the concept of state distinguishability. For the purpose of our work, in the following, we recall few relevant definitions \cite{Brunner14}. 

\begin{definition}[Perfect distinguishability]
Two states $\omega_1,\omega_2\in\Omega_S$ are perfectly distinguishable whenever there exists some measurement $\mathcal{M}=\{e_1,e_2\in\mathcal{E}_S~|~e_1+e_2=u\}$ such that $e_i(\omega_j)=\delta_{ij}$.
\end{definition}
For instance, two quantum states $\ket{\psi},\ket{\phi}\in\mathbb{C}^d$ are perfectly distinguishable if and only if they are orthogonal, a fact which follows from the seminal no-cloning theorem \cite{Wootters82}. On the other hand, in discrete classical probability theory the state spaces are simplices and any two extreme points are perfectly distinguishable \cite{Wilce21}.
\begin{definition}[Operational Dimension]
Operational dimension $\mathbb{O}(S)$ of a system $S$ is the maximum cardinality of the set of states $\Omega_n:=\{\omega_1,\cdots,\omega_n\}\subset\Omega_S$ such that all the states in $\Omega_n$ are perfectly distinguishable in a single measurement.
\end{definition}
For instance $\mathbb{O}(\mathbb{C}^d)=d$ although the the dimension of the vector space in which $\mathcal{D}(\mathbb{C}^d)$ is embedded is $d^2-1$. Operational dimension of a system quantifies its classical information carrying capacity \cite{Hardy01,Wootters86} (see also \cite{Muller12,Arai19}), {\it i.e.} sending a system with operational dimension $\mathbb{O}(S)$ through a noiseless channel a sender can sends $\log_2\mathbb{O}(S)$-bits of classical information to a receiver. 
\begin{definition}[Information Dimension]
The information dimension $\mathbb{I}(S)$ of a system $S$ is the maximum cardinality of the set of states $\Omega_n:=\{\omega_1,\cdots,\omega_n\}\subset\Omega_S$ such that all the states in $\Omega_n$ are pairwise perfectly distinguishable.
\end{definition}
Note that while defining $\mathbb{O}(S)$ a single measurement is allowed to distinguish the states in the set $\Omega_n$. On the other hand, $\mathbb{I}(S)$ deals with the pairwise distinguishability and for different pairs of states $\{\omega_i,\omega_j\}$ in $\Omega_n$, different measurements $\mathcal{M}_{ij}$ can be performed to distinguish the pairs. Therefore, it clearly follows that $\mathbb{I}(\star)\ge\mathbb{O}(\star)$ for an arbitrary GPT system, and accordingly one can define a quantity called dimension mismatch, $\Delta(\star):=\mathbb{I}(\star)-\mathbb{O}(\star)$. For classical and quantum systems it follows from simple arguments that both theses dimensions are equal. However, as shown in \cite{Brunner14}, for the hypothetical toy model of Box world ($\Box$) the information dimension is strictly greater than the operational dimension. While $\mathbb{I}(\Box)=4$, one has that $\mathbb{O}(\Box)=2.$ 

\section{Results}
As already mentioned, the state space of maximal composition strictly contains the quantum state space, {\it i.e.}, $\mathcal{D}(\mathbb{C}^{d_A}\otimes\mathbb{C}^{d_B})\subset\Omega_{\overline{\text{SEP}}}(\mathbb{C}^{d_A},\mathbb{C}^{d_B})$. In particular, an entanglement witness operator $W\notin\mathcal{D}(\mathbb{C}^{d_A}\otimes\mathbb{C}^{d_B})$, whereas $W\in\Omega_{\overline{\text{SEP}}}(\mathbb{C}^{d_A},\mathbb{C}^{d_B})$. Although the state space of $\overline{\text{SEP}}$ theory is bigger than the quantum state space, the `nonlocal strength' of the bipartite system $S^{AB}_{\overline{\text{SEP}}}$ is no more than  $S^{AB}_{Q}$. This follows from a generic result by Barnum {\it et al.} \cite{Barnum10}, where it is proved that any no-signaling bipartite input-output probability distribution $P(ab|xy)$ obtained from $S^{AB}_{\overline{\text{SEP}}}$ can also be obtained from $S^{AB}_{Q}$; here $a$ and $b$ denote Alice's and Bob's output corresponding to their respective inputs $x$ and $y$. For a state $W\in\Omega_{\overline{\text{SEP}}}(\mathbb{C}^{d_A},\mathbb{C}^{d_B})$ the correlation $P(ab|xy)$ is obtained as
\begin{align*}
P(ab|xy)&=\tr[W(\pi^a_x\otimes\pi^b_y)],\\
\pi^a_x\in\mathcal{P}(\mathbb{C}^{d_A}),~\sum_a\pi^a_x=\mathbf{1}_{d_A}~&\&~\pi^b_y\in\mathcal{P}(\mathbb{C}^{d_B}),~\sum_b\pi^b_y=\mathbf{1}_{d_B}.
\end{align*}
As pointed out in \cite{Acin10}, the result of Barnum {\it et al.} can be seen as follows. According to Choi–Jamio\l{}kowski (CJ) isomorphism \cite{Choi75,Jamiolkowski72}, any $W\in\Omega_{\overline{\text{SEP}}}(\mathbb{C}^{d_A},\mathbb{C}^{d_B})\setminus\mathcal{D}(\mathbb{C}^{d_A}\otimes\mathbb{C}^{d_B})$ can be written as $[\mathcal{I}\otimes\Lambda](\phi^+)$, where $\Lambda$ is a positive map, $\mathcal{I}$ is the identity map, and $\phi^+$ is the projector on maximally entangled state. Furthermore, any such witness can also be written as $[\mathcal{I}\otimes\Lambda_{tp}](\psi)$, where $\Lambda_{tp}$ is positive and trace-preserving and $\psi$ is a projector onto a pure bipartite state \cite{Horodecki06}. Therefore, we have
\begin{align*}
P(ab|xy)&=\tr[W(\pi^a_x\otimes\pi^b_y)]=\tr[[\mathcal{I}\otimes\Lambda_{tp}](\psi)(\pi^a_x\otimes\pi^b_y)]\\
&=\tr[\psi(\pi^a_x\otimes\Lambda_{tp}^\star[\pi^b_y])]=\tr[\psi(\pi^a_x\otimes\tilde{\pi}^b_y)].
\end{align*}
Here $\Lambda^\star$ is the adjoint map of $\Lambda$, and since the adjoint of a positive trace-preserving map is positive and unital, $\{\tilde{\pi}^b_y:=\Lambda_{tp}^\star[\pi^b_y]\}_b$ forms a valid quantum measurement. 

We will now proceed to show that the system $S^{AB}_{\overline{\text{SEP}}}$ can yield stronger that quantum correlation in the timelike domain. We will establish this with the help of a communication game introduced in \cite{Naik22} which we briefly recall below.    

{\it Pairwise distinguishability game $\mathcal{P}_D^{[n]}$}:-- The game involves two players (Alice and Bob) and a Referee. In each run of the game, the Referee provides a classical message $\eta$ to Alice, randomly chosen from some set of messages $\mathcal{N}$, where $|\mathcal{N}|:=n$. In the same run Bob is asked a question $\mathbb{Q}(\eta,\eta^\prime)$ -- whether the message given to Alice is $\eta$ or $\eta^\prime$, where $\eta^\prime\neq\eta$. The winning condition demands Bob answer all questions correctly. Alice can help Bob by sending some information about the message she received. It is not hard to see that perfect winning demands Alice to encode the message on the states of some physical system that are pairwise distinguishable. With this game we are now in a position to prove one of our main results. 
\begin{table}[h!]
\begin{tabular}{ |p{ 1.4cm}||p{ 1.4cm}|p{ 1.4cm}|p{ 1.4cm}|p{ 1.4cm}| }
\hline
&\diagbox[innerwidth=1.4cm, dir=SW,innerleftsep=.0cm,innerrightsep=.3cm]{\\~~ $ \Phi^+ $}{~\\$\overline{\Phi^+ }$}\cellcolor{green!15}&
\diagbox[innerwidth=1.4cm, dir=SW,innerleftsep=.0cm,innerrightsep=.3cm]{ \\~~ $ \Phi^- $}{~\\$\overline{\Phi^- }$}\cellcolor{green!15}&
\diagbox[innerwidth=1.4cm, dir=SW,innerleftsep=.0cm,innerrightsep=.3cm]{ \\~~ $ \Psi^+ $}{~\\ $\overline{\Psi^+ }$}\cellcolor{green!15}&
\diagbox[innerwidth=1.4cm, dir=SW,innerleftsep=.0cm,innerrightsep=.3cm]{ \\~~ $ \Psi^- $}{\\~$\overline{\Psi^- }$}\cellcolor{green!15}
\\
\hline\hline\hline
\diagbox[innerwidth=1.4cm, dir=SW,innerleftsep=.0cm,innerrightsep=.3cm]{\\~~ $ \Phi^+ $}{~\\$\overline{\Phi^+ }$}\cellcolor{green!15}&
$~~~~~\mathbf{NA}$ &$~A^y\otimes A^y$ &$~~~~\mathbf{1}\otimes \mathbf{1}$  &$~~~~\mathbf{1}\otimes \mathbf{1}$ 
\\
\hline
\diagbox[innerwidth=1.4cm, dir=SW,innerleftsep=.0cm,innerrightsep=.3cm]{ \\~~ $ \Phi^- $}{~\\$\overline{\Phi^- }$}\cellcolor{green!15}&$~A^y\otimes A^y$  &$~~~~~\mathbf{NA}$&$~~~~\mathbf{1}\otimes \mathbf{1}$ &$~~~~\mathbf{1}\otimes \mathbf{1}$ \\
\hline
\diagbox[innerwidth=1.4cm, dir=SW,innerleftsep=.0cm,innerrightsep=.3cm]{ \\~~ $ \Psi^+ $}{~\\ $\overline{\Psi^+ }$}\cellcolor{green!15}&$~~~~\mathbf{1}\otimes \mathbf{1}$  &$~~~~\mathbf{1}\otimes \mathbf{1}$&$~~~~~\mathbf{NA}$ &$~A^y\otimes A^y$\\
\hline
\diagbox[innerwidth=1.4cm, dir=SW,innerleftsep=.0cm,innerrightsep=.3cm]{ \\~~ $ \Psi^- $}{\\~$\overline{\Psi^- }$}\cellcolor{green!15}&$~~~~\mathbf{1}\otimes \mathbf{1}$  &$~~~~\mathbf{1}\otimes \mathbf{1}$&$~A^y\otimes A^y$&$~~~~~\mathbf{NA}$\\
\hline\hline\hline
\diagbox[innerwidth=1.4cm, dir=SW,innerleftsep=.0cm,innerrightsep=.3cm]{ \\~~ $\overline{\Phi^+ }$}{~ \\ ~} \cellcolor{brown!25}&$~A^x\otimes A^x$  &$~A^y\otimes A^y$&$~~~~\mathbf{1}\otimes \mathbf{1}$ &$~~~~\mathbf{1}\otimes \mathbf{1}$\\
\hline
\diagbox[innerwidth=1.4cm, dir=SW,innerleftsep=.0cm,innerrightsep=.3cm]{ \\~~ $\overline{\Phi^- }$}{~ \\ ~}
\cellcolor{brown!25}&$~A^y\otimes A^y$  &$~A^x\otimes A^x$ &$~~~~\mathbf{1}\otimes \mathbf{1}$ &$~~~~\mathbf{1}\otimes \mathbf{1}$\\
\hline
\diagbox[innerwidth=1.4cm, dir=SW,innerleftsep=.0cm,innerrightsep=.3cm]{ \\~~ $\overline{\Psi^+ }$}{~ \\ ~}\cellcolor{brown!25}&$~~~~\mathbf{1}\otimes \mathbf{1}$ &$~~~~\mathbf{1}\otimes \mathbf{1}$&$~A^x\otimes A^x$ &$~A^y\otimes A^y$\\
\hline

\diagbox[innerwidth=1.4cm, dir=SW,innerleftsep=.0cm,innerrightsep=.3cm]{ \\~~ $\overline{\Psi^- }$}{~ \\ ~}\cellcolor{brown!25}&$~~~~\mathbf{1}\otimes \mathbf{1}$  &$~~~~\mathbf{1}\otimes \mathbf{1}$&$~A^y\otimes A^y$ &$~A^x\otimes A^x$\\
\hline
\end{tabular}
\caption{\label{unitaries} The unitaries required to construct the measurement $\mathcal{M}[U\otimes V]$ for pairwise distinguishability of the states in $\$[8]$ are given. The states in the horizontal upper (lower) diagonal can be distinguished from the states in the vertical upper (lower) diagonal using the corresponding unitaries. For instance, the measurement to distinguish the pair $\{\Psi^+, \Psi^-\}$ as well as the pair $\{\overline{\Psi^+},\overline{\Psi^-}\}$ is given by the entry in the third row and fourth column, {\it i.e.}, $\mathcal{M}[A^y\otimes A^y]$, whereas the pair $\{\overline{\Psi^+},\Psi^+\}$ is distinguished by the measurement given in seventh row, third column, {\it i.e.}, $\mathcal{M}[A^x\otimes A^x]$. $\mathbf{NA}$ means that a state cannot be distinguished from itself.}
\end{table}
\begin{theorem}\label{theo1}
The game $\mathcal{P}_D^{[8]}$ cannot be won if Alice uses the system $[\mathbb{C}^2,\mathbb{C}^2,\bigotimes_Q]$ to encode her message whereas $[\mathbb{C}^2,\mathbb{C}^2,\bigotimes_{\overline{\text{SEP}}}]$ system yields a perfect winning strategy.  
\end{theorem}
\begin{proof}
Perfect winning of the game $\mathcal{P}_D^{[n]}$ requires Alice to communicate to Bob a physical system which has information dimension at least $n$. For the two-qubit system $[\mathbb{C}^2,\mathbb{C}^2,\bigotimes_Q]$, the information dimension is the same as its operational dimension which is $4$, and therefore, $\mathcal{P}_D^{[8]}$ game cannot be won perfectly by communicating two qubits. 

We now provide an explicit strategy to win the game $\mathcal{P}_D^{[8]}$ using the system $[\mathbb{C}^2,\mathbb{C}^2,\bigotimes_{\overline{\text{SEP}}}]$. Let Alice use the set  $\$[8]\equiv\left\{\Phi^\pm,\Psi^\pm,\overline{\Phi^\pm},\overline{\Psi^\pm}\right\}\subset\Omega_{\overline{\text{SEP}}}(\mathbb{C}^2,\mathbb{C}^2)$ of eight different states to encode her messages; where $\chi:=\ket{\chi}\bra{\chi}$, $\ket{\Phi^\pm}:=(\ket{00}\pm\ket{11})/\sqrt{2}$, $\ket{\Psi^\pm}:=(\ket{01}\pm\ket{10})/\sqrt{2}$, and 
$\overline{\chi}:=\mathcal{I}\otimes\mathrm{T}(\chi)$ with $\mathcal{I}$ denoting the identity map and $\mathrm{T}$ denoting the transposition map (in the computational basis). It remains to be shown that the states in $\$[8]$ are pairwise distinguishable with measurements constituted by the effects from the set $\mathcal{E}_{\overline{\text{SEP}}}(\mathbb{C}^2,\mathbb{C}^2)$. 

Consider the pair of states $\Phi^+$ and $\Psi^+$, and the measurement
\begin{align*}
\mathcal{M}\equiv\begin{cases}
E_{even}:=\ket{0}\bra{0}\otimes\ket{0}\bra{0}+\ket{1}\bra{1}\otimes\ket{1}\bra{1},\\
E_{odd}=\mathcal{I}-E_{even}:=\ket{0}\bra{0}\otimes\ket{1}\bra{1}\\\hspace{2.5cm}+\ket{1}\bra{1}\otimes\ket{0}\bra{0}.
\end{cases}
\end{align*}
Clearly, $\mathcal{M}$ is a valid measurement on the system $[\mathbb{C}^2,\mathbb{C}^2,\bigotimes_{\overline{\text{SEP}}}]$ as $E_{odd}, E_{even}\in\mathcal{E}_{\overline{\text{SEP}}}(\mathbb{C}^2,\mathbb{C}^2)$, where $E_{even}$ is the projector of even number of up spin and $E_{odd}$ is the projector of odd number of up spin. A straightforward calculation yields
\begin{align*}
\tr(\Phi^+E_{odd})=1,~~\tr(\Phi^+E_{even})=0;\\  
\tr(\Psi^+E_{odd})=0,~~\tr(\Psi^+E_{even})=1. 
\end{align*}
Therefore, the measurement $\mathcal{M}$ perfectly distinguishes the states $\Phi^+$ and $\Psi^+$. To show the same for any pair of states in $\$[8]$, let us denote as $\mathcal{M}[U\otimes V]$ the measurement obtained from $\mathcal{M}$ through the unitary rotation $U\otimes V$, {\it i.e.}  $\mathcal{M}[U\otimes V]:=\{U\otimes V E_{odd}U^\dagger\otimes V\dagger,U\otimes V E_{even}U^\dagger\otimes V\dagger\}$. As shown in Table \ref{unitaries}, choosing $U$ and $V$ appropriately from the set $\left\{\mathbf{1}:=\begin{pmatrix} 1 & 0\\0 & 1\\\end{pmatrix},A^x:=\frac{1}{\sqrt{2}}\begin{pmatrix}1 & -i\\-i & 1\\\end{pmatrix},A^y:=\frac{1}{\sqrt{2}}\begin{pmatrix}1 & -1\\1 & 1\\\end{pmatrix}\right\}$ any pair of states in $\$[8]$ can be distinguished perfectly by the measurement $\mathcal{M}[U\otimes V]$ . This completes the proof.
\end{proof}
Theorem \ref{theo1} thus establishes that $\overline{\text{SEP}}$ composition of two elementary qubits can result in a correlation that can't be achieved with two qubits quantum composition. As an immediate corollary we have a lower bound on the information dimension of the system $[\mathbb{C}^2,\mathbb{C}^2,\bigotimes_{\overline{\text{SEP}}}]$.  
\begin{corollary}\label{coro1}
The information dimension of the system $[\mathbb{C}^2,\mathbb{C}^2,\bigotimes_{\overline{\text{SEP}}}]$ is at least $8$, {\it i.e.} $\mathbb{I}[\mathbb{C}^2,\mathbb{C}^2,\bigotimes_{\overline{\text{SEP}}}]\ge8$. 
\end{corollary}
At present we do not know whether the above bound is tight and leave this question open for future research. Rather, we proceed to find the operational dimension of the systems obtained through the $\overline{\text{SEP}}$ composition. To this aim, we first prove the following proposition. 
\begin{proposition}\label{prop1}
Every POPT state $W_{AB}\in\Omega_{\overline{\text{SEP}}}(\mathbb{C}^{d_A},\mathbb{C}^{d_B})$ can be written as $(\mathcal{I}_A\otimes\Lambda_{R \to B})(\rho_{AR})$ where $\Lambda$ is a positive, unital map and $\rho_{AR}\in\mathbb{C}^{d_A}\otimes\mathbb{C}^{d_R}$ is a pure quantum state independent of $W_{AB}$.
\end{proposition}
\begin{proof}
We start by defining the following:\footnote{The techniques used in the proof are motivated from \cite{Horodecki06}.}
\begin{align*}
W^\prime_{AB}&:= W_{AB} + \mathbf{1}_A \otimes P_B^{\perp};~~~~P_B^{\perp}:= \mathbf{1}_B - P_B;\\
P_B &:= \text{Projector onto the support of } W_B;\\
W_B &:= \tr_A\left(W_{AB}\right);~~W^\prime_B := \tr_A\left(W^\prime_{AB}\right).
\end{align*}
$W'_{AB}$ is a POPT state (unnormalized) as for any separable effect $\pi_A \otimes \pi_B$ we have, $\tr{[(W'_{AB})(\pi_A \otimes \pi_B)]}= \tr{[(W_{AB})(\pi_A \otimes \pi_B)]} + \tr{(\pi_A)} \tr{(P_B^{\perp} \pi_B)}\ge0$. Thus, $W'_B$ is a full rank positive operator, and hence we can define:
\begin{align*}
W''_{AB}&:= \left[\mathbf{1}_A \otimes \left(W'_B\right)^{-1/2}\right] W'_{AB} \left[\mathbf{1}_A \otimes \left(W'_B\right)^{-1/2}\right],
\end{align*}
$\left(W'_B\right)^{-1/2}$, being a positive operator, implies that $W''_{AB}$ is a POPT: $\tr{[(W''_{AB})(\pi_A \otimes \pi_B)]}= \tr{\left[W'_{AB}\left(\pi_A \otimes \left(W'_B\right)^{-1/2} \pi_B \left(W'_B\right)^{-1/2} \right)\right]}\ge 0$. Further, $\tr_A{(W''_{AB})}= \sum_i \left(W'_B\right)^{-1/2}\bra{i}_AW'_{AB}\ket{i}_A \left(W'_B\right)^{-1/2}=  \left(W'_B\right)^{-1/2}W'_B\left(W'_B\right)^{-1/2} = \mathbf{1}_B$.

Using CJ isomorphism we can write
$W''_{AB} = \mathcal{I}_A \otimes \mathcal{U}_{S \to B} (\ket{\chi^+}_{AS}\bra{\chi^+})$, where $\ket{\chi^+}_{AB} := \sum_{i}{\ket{i}_{A} \ket{i}_B }$ is the unnormalized maximally entangled state and $\mathcal{U}_{S \to B}$ is a positive map. More explicitly, the action of  $\mathcal{U}_{S \to B}$ is given by,  $\mathcal{U}_{S \to B}(M_S) = \tr_S{[(M_S^{T} \otimes \mathbf{1}_B)(W''_{SB})]}$. $\mathcal{U}_{S \to B}$ is unital, since   $\mathcal{U}_{S \to B} (\mathbf{1}_{S})=\tr_S({W''_{SB}})=\mathbf{1}_B$. Furthermore, it is easy to check that $W_{AB} = (\mathbf{1}_A \otimes V_B^\dagger) W''_{AB} (\mathbf{1}_A \otimes V_B)$, where $V_B:= \left(W'_B\right)^{1/2} P_B$.

Let us now define a new completely positive, unital map 
\begin{align*}
\mathcal{Y}_{BC \to B}(M_{BC}):= V_B^{\dagger}  \bra{0}_C M\ket{0}_C V_B   + V_B^{'\dagger} \bra{1}_C M\ket{1}_C V_B^{'}~, 
\end{align*}
where $V_B^{'}$ is chosed so that it satisfies the condition $V_B^{\dagger} V_B + V_B^{'\dagger} V_B^{'} = \mathbf{1}_B$ and $\mathcal{H}_C := \mathbb{C}^2 $. The above map is the adjoint of the completely positive, trace preserving map having the Kraus operators $\{V_B \otimes \ket{0}_C , V_B^{'} \otimes \ket{1}_C\}$; and it has the property,
\begin{align*}
    \mathcal{I}_A \otimes \mathcal{Y}_{BC \to B}&(M_{AB} \otimes \ket{0}_C\bra{0}) \\
    &= (\mathbf{1}_A \otimes V_B^\dagger) M_{AB} (\mathbf{1}_A \otimes V_B).
\end{align*}
This further leads us to,
\begin{align*}
&W_{AB}= \left(\mathbf{1}_A \otimes V_B^\dagger\right) W''_{AB}\left(\mathbf{1}_A \otimes V_B\right)\\
&= \mathcal{I}_A \otimes \mathcal{Y}_{BC \to B}\left(W''_{AB} \otimes \ket{0}_C\bra{0}\right)\\
&= \mathcal{I}_A \otimes \mathcal{Y}_{BC \to B}\left[\left(\mathcal{I}_A \otimes \mathcal{U}_{S \to B}\right)\left(\ket{\chi^+}_{AS}\bra{\chi^+}\right) \otimes \ket{0}_C\bra{0} \right]\\
&= \left(\mathcal{I}_A \otimes \mathcal{Y}_{BC \to B }\right)\circ\left(\mathcal{I}_A \otimes \mathcal{U}_{S \to B} \otimes  \mathcal{I}_C\right)\left[\ket{\chi^+}_{AS}\bra{\chi^+} \otimes \ket{0}_C\bra{0}\right]\\
&= \mathcal{I}_A \otimes \left(\mathcal{Y}_{BC \to B }\circ \mathcal{U'}_{SC \to BC}\right) \left[\ket{\chi^+}_{AS}\bra{\chi^+} \otimes \ket{0}_C\bra{0}\right],
\end{align*}
where $\mathcal{U'}_{SC \to BC} := \mathcal{U}_{S \to B} \otimes  \mathcal{I}_C $. Let $d_S := \dim(\mathcal{H}_S)$, $\Lambda_{SC \to B} := \frac{1}{d_S} \mathcal{Y}_{BC \to B} \circ \mathcal{U'}_{SC \to BC}$, $\ket{\psi}_{ASC}:= \frac{1}{\sqrt{d_S}} \ket{\chi^+}_{AS}\ket{0}_C$, and $\mathcal{H}_R := \mathcal{H}_S \otimes \mathcal{H}_C$. Thus we have,
\begin{align*}
W_{AB}&= (\mathcal{I}_A \otimes \Lambda_{R \to B}) (\ket{\psi}_{AR}\bra{\psi}),
\end{align*}
where $\Lambda_{R \to B}$ is the composition of a completely positive unital map and a positive  unital map; and therefore it is positive and unital. This completes the proof.
\end{proof}
We are now in a position to prove another important result of this work. 
\begin{theorem}\label{theo2}
The Operational Dimension of the system $[\mathbb{C}^{d_A},\mathbb{C}^{d_B},\bigotimes_{\overline{\text{SEP}}}]$ is $d_A d_B$.  
\end{theorem}
\begin{proof}
The proof is similar in spirit to Lemma 24 of Ref.\cite{Muller12}. However, while the assumption of `transitivity' is used there, here we use the Proposition \ref{prop1}.

Let the Operational dimension of the system $[\mathbb{C}^{d_A},\mathbb{C}^{d_B},\bigotimes_{\overline{\text{SEP}}}]$ be $N$. $N$ must be lower bounded by $d:=d_A d_B$, as there exists $d$ number of quantum states that can be perfectly distinguished by a single separable measurement. For instance, the set of pure states $\{\ket{ij}~|~i=1,\cdots d_A~ \&~j=1,\cdots d_B\}$ can be distinguished by the separable measurement $\{\ketbra{i}{i} \otimes \ketbra{j}{j}\}_{i,j=1}^{d_A,d_B}$. As we have considered the operational dimension of the system $[\mathbb{C}^{d_A},\mathbb{C}^{d_B},\bigotimes_{\overline{\text{SEP}}}]$ to be $N$, there must exist $N$ POPT states $\{W_1, \cdots, W_N\}$ and a separable measurement  $\{E_1,\cdots,E_N~|~\sum_{i=1}^{N} E_i =\mathbf{1}_{AB}\}$ such that $Tr(E_i W_j)=\delta_{ij} ~;~\forall i,j$. According to Proposition \ref{prop1}, $\forall~j,~W_j = (\mathcal{I}\otimes\Lambda_j)(\rho)$ for some positive, unital map $\Lambda_j \colon \mathcal{L}(\mathcal{H}_R) \to \mathcal{L}(\mathcal{H}_B)$ and pure state $\rho \in \mathcal{L}(\mathcal{H}_{A} \otimes \mathcal{H}_{R})$. Denoting the projector on the orthogonal support of $\rho$ as $P := \mathbf{1}_{AR} - \rho$, we have 
\begin{align*}
d &=  Tr(\mathbf{1}_{AB}) = \sum_{i=1}^{N} Tr(E_i)= \sum_{i=1}^{N} Tr[E_i (\mathcal{I}\otimes\Lambda_i)(\mathbf{1}_{AR})]\nonumber\\
&= \sum_{i=1}^{N} Tr[E_i (\mathcal{I}\otimes\Lambda_i)( \rho + P )]\nonumber\\
&= \sum_{i=1}^{N} Tr[E_i (\mathcal{I}\otimes\Lambda_i)(\rho )] + \sum_{i=1}^{N} Tr[E_i (\mathcal{I}\otimes\Lambda_i)( P )]\nonumber\\
&= \sum_{i=1}^{N} Tr[E_i W_i] + \sum_{i=1}^{N} Tr[(\mathcal{I}\otimes\Lambda_i^*)(E_i)  P ],   \end{align*}
where $\Lambda_i^*$ is the adjoint map of $\Lambda_i$ and hence positive. Furthermore, $(\mathbb{1}\otimes\Lambda_i^*)(E_i)$ are positive operators since $E_i$'s are separable. Therefore we have,
\begin{align*}
d &\ge  \sum_{i=1}^{N} Tr[E_i W_i]= \sum_{i=1}^{N} \delta_{ii}=N. \label{eq3}
\end{align*}
Since we know that $N\ge d$, therefore we conclude $N=d$. This completes the proof.
\end{proof}
While in Theorem \ref{theo1} we have shown that the $\overline{\text{SEP}}$ composition of two elementary qubits can yield stronger timelike correlation than their quantum composition ({\it i.e.}, two-qubit), Theorem \ref{theo2} establishes that such a composition is not strong enough to show super-additive feature of the information carrying capacity \cite{Self1}. In this regard, a more generic result is presented in the next proposition.  

\begin{proposition}\label{prop2}
The operational dimension of any bipartite composition (with normalized state space denoted as $\Omega_{AB}$) of two elementary quantum systems $\mathbb{C}^{d_A}$ and $\mathbb{C}^{d_B}$ is $d_Ad_B$  if $(\mathbf{1}_{AB}-W_{AB})$ lies within the unnormalised state cone $\forall ~ W_{AB} \in \Omega_{AB}$. 
\end{proposition}
\begin{proof}
Let $N$ be the operational dimension of the composite system. Then, there must exist a set containing $N$ states $\{W_1, \cdots, W_N\}$ and measurement $\{E_1,\cdots,E_N~|~\sum_{i=1}^{N} E_i =\mathbf{1}_{AB}\}$ such that $Tr(E_i W_j)=\delta_{ij}$, $\forall~i,j$. Manifestly it follows that $N\ge d:=d_Ad_B$, since $\Omega_{\text{SEP}}\subseteq\Omega_{AB}\subseteq\Omega_{\overline{\text{SEP}}}$ and $\mathcal{E}_{\overline{\text{SEP}}}\subseteq\mathcal{E}_{AB}\subseteq\mathcal{E}_{\text{SEP}}$. As argued in the proof of Theorem \ref{theo2}, there always exist $d_Ad_B$ number of product states that can be perfectly distinguished by a separable measurement. On the other hand,
\begin{align*}
\sum_{i=1}^{N} Tr(E_i (\mathbf{1} - W_i)) &= \sum_{i=1}^{N} Tr(E_i) - \sum_{i=1}^{N} Tr(E_i W_i) \\
&= d - \sum_{i=1}^{N} \delta_{ii} = d - N.
\end{align*}
Since $(\mathbf{1}_{AB}-W_{AB})$ is an unnormalised state by assumption, $d-N \ge 0$ or $d\ge N$, which completes the proof.
\end{proof}
While Proposition \ref{prop2} assumes elementary systems to be quantum it can however be further generalized within the GPT framework. 
\begin{proposition}\label{prop3}
Let $\mathcal{S}_{AB}\equiv(\Omega_{AB},\mathcal{E}_{AB})$ be a composite systems consisting two elementary systems $\mathcal{S}_{A}\equiv (\Omega_{A},\mathcal{E}_{A})$ and $\mathcal{S}_{B}\equiv (\Omega_{B},\mathcal{E}_{B})$ with operational dimensions $N_A$ and $N_B$ respectively. The operational dimension of $\mathcal{S}_{AB}$ is $N_A N_B$ if $~\exists~\omega'_{AB} \in \Omega_{AB}$ such that $(N_A N_B \omega'_{AB} - \omega_{AB})$ is an unnormalized state $\forall ~ \omega_{AB} \in \Omega_{AB}$.
\end{proposition}
\begin{proof}
The operational dimension $N_{AB}$ of the composite system $\mathcal{S}_{AB}$ is always greater than the product of the operational dimension of the elementary systems, {\it i.e.}, $N_{AB}\ge N_{A}N_{B}$. This simply follows form the fact that any valid composition includes the products states and product effects in its description. Now we have,
\begin{align*}
\sum_{i=1}^{N_{AB}} e_i&(N_A N_B \omega' - \omega_i) = N_A N_B \sum_{i=1}^{N_{AB}} e_i( \omega')- \sum_{i=1}^{N_{AB}} e_i( \omega_i)\\
&= N_A N_B u( \omega')- \sum_{i=1}^{N_{AB}} \delta_{ii}=N_A N_B - N_{AB}\ge0,
\end{align*}
since by assumption, $(N_A N_B \omega' - \omega_i)$ is an unnormalised state $\forall ~ \omega_i$. Therefore, $N_{AB}\le N_A N_B$, which completes the proof.
\end{proof}
Form Theorem \ref{theo1} and Theorem \ref{theo2} we can conclude that $\Delta[\mathbb{C}^2,\mathbb{C}^2,\bigotimes_{\overline{\text{SEP}}}]\ge 4$, where $\Delta$ refers to the dimension mismatch of the theory. On the other hand, we can also conclude that the gap between the information dimensions of the systems $[\mathbb{C}^2,\mathbb{C}^2,\bigotimes_{\overline{\text{SEP}}}]$ and $[\mathbb{C}^2,\mathbb{C}^2,\bigotimes_{Q}]$ is at least $4$, {\it i.e.},
\begin{align*}
\mathbb{I}\left[\mathbb{C}^2,\mathbb{C}^2,\otimes_{\overline{\text{SEP}}}\right]-\mathbb{I}\left[\mathbb{C}^2,\mathbb{C}^2,\otimes_{Q}\right]\ge4.
\end{align*}
Our next result shows that this gap can be increased further by considering more number of elementary systems.
\begin{table}[t!]
\begin{tabular}{|p{1.5cm}||p{3cm}|p{3cm}|}
\hline
~~~~Measu. &~Odd no. 'up' spin  &~Even no. 'up' spin \\
\hline\hline
& $\Phi^{-}_{000},\Phi^{-}_{001},
\Phi^{+}_{010},\Phi^{+}_{011}$ 

&$\Phi^{+}_{000},\Phi^{+}_{001},
\Phi^{-}_{010},\Phi^{-}_{011}$\\
$~~~(y,y,x)$

&$\overline{\Phi^{+}_{000}},\overline{\Phi^{+}_{001}},\overline{\Phi^{-}_{010}},\overline{\Phi^{-}_{011}}$ 
& $\overline{\Phi^{-}_{000}},\overline{\Phi^{-}_{001}},\overline{\Phi^{+}_{010}},\overline{\Phi^{+}_{011}}$\\
& $\overline{\overline{\Phi^{+}_{000}}},\overline{\overline{\Phi^{+}_{001}}},\overline{\overline{\Phi^{-}_{010}}},\overline{\overline{\Phi^{-}_{011}}}$ 
& $\overline{\overline{\Phi^{-}_{000}}},\overline{\overline{\Phi^{-}_{001}}},\overline{\overline{\Phi^{+}_{010}}},\overline{\overline{\Phi^{+}_{011}}}$\\
&&\\
\hline
&$\Phi^{-}_{000},\Phi^{+}_{001},
\Phi^{-}_{010},\Phi^{+}_{011}$ 
&$\Phi^{+}_{000},\Phi^{-}_{001},
\Phi^{+}_{010},\Phi^{-}_{011}$\\
&&\\
$~~~(y,x,y)$
&$\overline{\Phi^{-}_{000}},\overline{\Phi^{+}_{001}},\overline{\Phi^{-}_{010}},\overline{\Phi^{+}_{011}}$ 
&$\overline{\Phi^{+}_{000}},\overline{\Phi^{-}_{001}},\overline{\Phi^{+}_{010}},\overline{\Phi^{-}_{011}}$\\
&&\\
&$\overline{\overline{\Phi^{+}_{000}}},\overline{\overline{\Phi^{-}_{001}}},\overline{\overline{\Phi^{+}_{010}}},\overline{\overline{\Phi^{-}_{011}}}$ 
&$\overline{\overline{\Phi^{-}_{000}}},\overline{\overline{\Phi^{+}_{001}}},\overline{\overline{\Phi^{-}_{010}}},\overline{\overline{\Phi^{+}_{011}}}$\\
&&\\
\hline
&$\Phi^{+}_{000},\Phi^{+}_{001},
\Phi^{+}_{010},\Phi^{+}_{011}$ 
&$\Phi^{-}_{000},\Phi^{-}_{001},
\Phi^{-}_{010},\Phi^{-}_{011}$\\
&&\\
$~~~(x,x,x)$
&$\overline{\Phi^{+}_{000}},\overline{\Phi^{+}_{001}},\overline{\Phi^{+}_{010}},\overline{\Phi^{+}_{011}}$ 
&$\overline{\Phi^{-}_{000}},\overline{\Phi^{-}_{001}},\overline{\Phi^{-}_{010}},\overline{\Phi^{-}_{011}}$\\
&&\\
& $\overline{\overline{\Phi^{+}_{000}}},\overline{\overline{\Phi^{+}_{001}}},\overline{\overline{\Phi^{+}_{010}}},\overline{\overline{\Phi^{+}_{011}}}$
&$\overline{\overline{\Phi^{-}_{000}}},\overline{\overline{\Phi^{-}_{001}}},\overline{\overline{\Phi^{-}_{010}}},\overline{\overline{\Phi^{-}_{011}}}$\\
&&\\
\hline
&$\Phi^{-}_{000},\Phi^{+}_{001},
\Phi^{+}_{010},\Phi^{-}_{011}$ 
&$\Phi^{+}_{000},\Phi^{-}_{001},
\Phi^{-}_{010},\Phi^{+}_{011}$\\
&&\\
$~~~(x,y,y)$
&$\overline{\Phi^{+}_{000}},\overline{\Phi^{-}_{001}},\overline{\Phi^{-}_{010}},\overline{\Phi^{+}_{011}}$ 
& $\overline{\Phi^{-}_{000}},\overline{\Phi^{+}_{001}},\overline{\Phi^{+}_{010}},\overline{\Phi^{-}_{011}}$\\
&&\\
&$\overline{\overline{\Phi^{-}_{000}}},\overline{\overline{\Phi^{+}_{001}}},\overline{\overline{\Phi^{+}_{010}}},\overline{\overline{\Phi^{-}_{011}}}$ 
& $\overline{\overline{\Phi^{+}_{000}}},\overline{\overline{\Phi^{-}_{001}}},\overline{\overline{\Phi^{-}_{010}}},\overline{\overline{\Phi^{+}_{011}}}$\\
&&\\
\hline
& $\Phi^{+}_{000},\Phi^{-}_{000},
\Phi^{+}_{011},\Phi^{-}_{011}$ 
&$\Phi^{+}_{001},\Phi^{-}_{001},
\Phi^{+}_{010},\Phi^{-}_{010}$\\
&&\\
$~~~(y,z,z)$
&$\overline{\Phi^{+}_{000}},\overline{\Phi^{-}_{000}},\overline{\Phi^{+}_{011}},\overline{\Phi^{-}_{011}}$ 
& $\overline{\Phi^{+}_{001}},\overline{\Phi^{-}_{001}},\overline{\Phi^{+}_{010}},\overline{\Phi^{-}_{010}}$\\
&&\\

&$\overline{\overline{\Phi^{+}_{000}}},\overline{\overline{\Phi^{-}_{000}}},\overline{\overline{\Phi^{+}_{011}}},\overline{\overline{\Phi^{-}_{011}}}$
& $\overline{\overline{\Phi^{+}_{001}}},\overline{\overline{\Phi^{-}_{001}}},\overline{\overline{\Phi^{+}_{010}}},\overline{\overline{\Phi^{-}_{010}}}$\\
&&\\
\hline
& $\Phi^{+}_{000},\Phi^{-}_{000},
\Phi^{+}_{001},\Phi^{-}_{011}$ 
&$\Phi^{+}_{010},\Phi^{-}_{010},
\Phi^{+}_{011},\Phi^{-}_{011}$\\
&&\\
$~~~(z,z,y)$
&$\overline{\Phi^{+}_{000}},\overline{\Phi^{-}_{000}},\overline{\Phi^{+}_{001}},\overline{\Phi^{-}_{001}}$ 
& $\overline{\Phi^{+}_{010}},\overline{\Phi^{-}_{010}},\overline{\Phi^{+}_{011}},\overline{\Phi^{-}_{011}}$\\
&&\\
& $\overline{\overline{\Phi^{+}_{000}}},\overline{\overline{\Phi^{-}_{000}}},\overline{\overline{\Phi^{+}_{001}}},\overline{\overline{\Phi^{-}_{001}}}$
& $\overline{\overline{\Phi^{+}_{010}}},\overline{\overline{\Phi^{-}_{010}}},\overline{\overline{\Phi^{+}_{011}}},\overline{\overline{\Phi^{-}_{011}}}$\\
&&\\
\hline
& $\Phi^{+}_{000},\Phi^{-}_{000},
\Phi^{+}_{010},\Phi^{-}_{010}$ 
&$\Phi^{+}_{001},\Phi^{-}_{001},
\Phi^{+}_{011},\Phi^{-}_{011}$\\
&&\\
$~~~(z,y,z)$
&$\overline{\Phi^{+}_{000}},\overline{\Phi^{-}_{000}},\overline{\Phi^{+}_{010}},\overline{\Phi^{-}_{010}}$ 
&$\overline{\Phi^{+}_{001}},\overline{\Phi^{-}_{001}},\overline{\Phi^{+}_{011}},\overline{\Phi^{-}_{011}}$\\
&&\\
&$\overline{\overline{\Phi^{+}_{000}}},\overline{\overline{\Phi^{-}_{000}}},\overline{\overline{\Phi^{+}_{010}}},\overline{\overline{\Phi^{-}_{010}}}$ 
& $\overline{\overline{\Phi^{+}_{001}}},\overline{\overline{\Phi^{-}_{001}}},\overline{\overline{\Phi^{+}_{011}}},\overline{\overline{\Phi^{-}_{011}}}$\\
&&\\
\hline
\end{tabular}
\caption{Pairwise distinguishablity of the set $\$[24]$. Using a particular separable measurement given in the first column, any state on the odd no. 'up' spin (second) column can be distinguished from any state on the even no. 'up' spin (third) column. For instance, the pair $\{\overline{\overline{\Phi^+_{000}}},\overline{\overline{\Phi^+_{001}}}\}$ (last row) is perfectly distinguishable via the separable measurement consisting of POVM elements given by  $\mathcal{M}\equiv \{E_{odd}, E_{even}\}$, where $E_{odd}$ and $E_{even}$ are rank four projectors comprising odd number of spin up and even number of spin up outcomes respectively for the Pauli measurement $(\sigma_Z, \sigma_Y, \sigma_Z)\equiv (z,y,z)$.}\label{tab2}
\end{table}

\begin{theorem} \label{theo3}
The Information Dimension of the system $[\mathbb{C}^{2},\mathbb{C}^{2},\mathbb{C}^{2},\bigotimes_{\overline{\text{SEP}}}]$ is at least $24$.
\end{theorem}
\begin{proof}
The proof is constructive and similar to the proof of Theorem \ref{theo1}. Consider the set of following $24$ states
\begin{align*}
\$[24]&:=\left\{\chi,\overline{\chi},\overline{\overline{\chi}}\right\}\subset\Omega_{\overline{\text{SEP}}}[\mathbb{C}^{2},\mathbb{C}^{2},\mathbb{C}^{2}];\\
\overline{\chi}&:=\mathcal{I}\tens{}\mathrm{T}\tens{}\mathcal{I}(\chi);\\
\overline{\overline{\chi}}&:=\mathcal{I}\tens{}\mathrm{T}\tens{}\mathrm{T}(\chi);\\
\ket{\chi}\in\Bigg\{\ket{\Phi^{\pm}_{000}}&:=\frac{1}{\sqrt{2}}(\ket{000}\pm\ket{111}),\\
\ket{\Phi^{\pm}_{001}}&:=\frac{1}{\sqrt{2}}(\ket{001}\pm\ket{110}),\\
\ket{\Phi^{\pm}_{010}}&:=\frac{1}{\sqrt{2}}(\ket{010}\pm\ket{101}),\\
\ket{\Phi^{\pm}_{011}}&:=\frac{1}{\sqrt{2}}(\ket{011}\pm\ket{100})\Bigg\}.
\end{align*}
We aim to show that the states in $\$[24]$ are pairwise distinguishable by fully separable measurements of the form 
\begin{align*}
\mathcal{M}\equiv\begin{cases}
E_{odd}~:=~~\ket{m}\bra{m}\otimes\ket{n}\bra{n}\otimes\ket{p}\bra{p}\\~~~~~~~~~~~~+\ket{m}\bra{m}\otimes\ket{n^{\perp}}\bra{n^{\perp}}\otimes\ket{p^{\perp}}\bra{p^{\perp}}\\~~~~~~~~~~~~+\ket{m^{\perp}}\bra{m^{\perp}}\otimes\ket{n}\bra{n}\otimes\ket{p^{\perp}}\bra{p^{\perp}}\\~~~~~~~~~~~~+\ket{m^{\perp}}\bra{m^{\perp}}\otimes\ket{n^{\perp}}\bra{n^{\perp}}\otimes\ket{p}\bra{p},\\
E_{even}=\mathbin{1}-E_{odd}\\~~~~~~~~~:=~~\ket{m}\bra{m}\otimes\ket{n}\bra{n}\otimes\ket{p^{\perp}}\bra{p^{\perp}}\\~~~~~~~~~~~~+\ket{m}\bra{m}\otimes\ket{n^{\perp}}\bra{n^{\perp}}\otimes\ket{p}\bra{p}\\~~~~~~~~~~~~+\ket{m^{\perp}}\bra{m^{\perp}}\otimes\ket{n}\bra{n}\otimes\ket{p}\bra{p}\\~~~~~~~~~~~~+\ket{m^{\perp}}\bra{m^{\perp}}\otimes\ket{n^{\perp}}\bra{n^{\perp}}\otimes\ket{p^{\perp}}\bra{p^{\perp}}
\end{cases}
\end{align*}
where  $\ket{r},\ket{r^{\perp}}$ are `up' and `down' eigenstates of spin measurement ($\hat{r}\cdot\sigma$) along the $\hat{r}$ direction, for $\hat{r}\in\{\hat{m},\hat{n},\hat{p}\}$. $E_{odd}$ comprises of odd number of `up spin' and $E_{even}$ comprises of even number of `up spin'. Clearly $\mathcal{M}$ is a an allowed measurement as $E_{odd},E_{even}\in\mathcal{E}_{\overline{\text{SEP}}}(\mathbb{C}^2,\mathbb{C}^2,\mathbb{C}^2)$. The required $m,n$ and $p$ to distinguish between different pairs of states are given on the  first column of Table \ref{tab2}. For instance, the pair of states  $\left\{\Phi^+_{000},\Phi^-_{000}\right\}$ can be perfectly distinguished by choosing $(m,n,p) = (y,y,x)$.  The measurement $\mathcal{M}_{\{\Phi^+_{000},\Phi^-_{000}\}}\equiv\{E_{odd},E_{even}\}$ is given by,
\begin{align*}
E_{odd}&:=~~~\ket{y}\bra{y}\otimes\ket{y}\bra{y}\otimes\ket{x}\bra{x}\\&~~~~~+\ket{y}\bra{y}\otimes\ket{y^{\perp}}\bra{y^{\perp}}\otimes\ket{x^{\perp}}\bra{x^{\perp}}\\&~~~~~+\ket{y^{\perp}}\bra{y^{\perp}}\otimes\ket{y}\bra{y}\otimes\ket{x^{\perp}}\bra{x^{\perp}}\\&~~~~~+\ket{y^{\perp}}\bra{y^{\perp}}\otimes\ket{y^{\perp}}\bra{y^{\perp}}\otimes\ket{x}\bra{x},\\
E_{even}&=\mathbin{1}-E_{odd}\\
&:=\ket{y}\bra{y}\otimes\ket{y}\bra{y}\otimes\ket{x^{\perp}}\bra{x^{\perp}}\\&~~~~~+\ket{y}\bra{y}\otimes\ket{y^{\perp}}\bra{y^{\perp}}\otimes\ket{x}\bra{x}\\&~~~~~+\ket{y^{\perp}}\bra{y^{\perp}}\otimes\ket{y}\bra{y}\otimes\ket{x}\bra{x}\\&~~~~~+\ket{y^{\perp}}\bra{y^{\perp}}\otimes\ket{y^{\perp}}\bra{y^{\perp}}\otimes\ket{x^{\perp}}\bra{x^{\perp}}.
\end{align*}
A straightforward calculation yields
\begin{align*}
\tr(\Phi^-_{000}E_{odd})=1,~~\tr(\Phi^-_{000}E_{even})=0;\\  
\tr(\Phi^+_{000}E_{odd})=0,~~\tr(\Phi^+_{000}E_{even})=1. 
\end{align*}
Therefore, the measurement $\mathcal{M}_{\{\Phi^+_{000},\Phi^-_{000}\}}$  perfectly distinguishes the states $\Phi^-_{000}$ and $\Phi^+_{000}$. As we show in Table \ref{tab2}, any pair of states in $\$[24]$ can be distinguished perfectly by such a measurement. This completes the proof.
\end{proof}
Theorem \ref{theo3} thus establishes that the $\mathcal{P}^{[24]}_D$ game can be won with three elementary qubits if the $\overline{\text{SEP}}$ composition is considered among them, whereas if we consider quantum composition five elementary qubits are required.

{\it Discussion:} Exploring theories other than quantum mechanics helps us to compare and contrast the information processing capabilities of quantum theory with other theories and gain new insights about the origin of such capabilities. In this work, we assume sub-systems to be quantum and ask how the information processing capabilites of composite systems  change when one uses different mathematical structures to describe composition. While  Barnum, et al. in \cite{Barnum10}  have shown that in the space-like scenario, bipartite maximal tensor product structure of local quantum systems cannot generate beyond quantum correlations, the authors in \cite{Lobo21} have shown that in a generalized Bell scenario every beyond quantum state can produce beyond quantum correlations. In this work, we have used a different approach wherein timelike scenarios are considered instead of the traditional spacelike Bell scenarios. We have provided concrete results which can be experimentally verified and be used as principles to single out the quantum composition rule. While Corollary \ref{coro1} and Theorem \ref{theo2} establish that the phenomenon of dimension mismatch occurs in $\overline{SEP}$ composition, it has been shown \cite{Arai19,Naik22} that dimension mismatch occurs in $SEP$ composition as well. A natural question then is to ask what other compositions can be ruled out using dimension mismatch. Another interesting direction to explore is by relaxing the assumption of quantum sub-systems. Propositions \ref{prop2} and \ref{prop3} provide some preliminary results in the GPT framework which may be useful in this regard. Our study forms an important piece of the quantum reconstruction program in which we seek to derive quantum theory from physical principles \cite{Hardy01, Barrett07, Chiribella11}.

\begin{acknowledgements}
MA and MB acknowledge funding from the National Mission in Interdisciplinary Cyber-Physical systems from the Department of Science and Technology through the I-HUB Quantum Technology Foundation (Grant no: I-HUB/PDF/2021-22/008). MB acknowledges support through the research grant of INSPIRE Faculty fellowship from the Department of Science and Technology, Government of India and the start-up research grant from SERB, Department of Science and Technology (Grant no: SRG/2021/000267).
\end{acknowledgements}


\begin{thebibliography}{99}
\bibitem{Self1} Information carrying capacity of a composition of two elementary systems with operational dimensions $d_A$ and $d_B$, respectively will exhibit super-additive feature whenever its operational dimension $N>d_Ad_B$, which accordingly results in $\log N>\log d_A+\log d_B$.
\end{thebibliography}
\end{document}